\documentclass[a4paper, runningheads]{llncs}
\usepackage{hyperref}
\usepackage{graphicx}
\usepackage{amsmath,amssymb}
\usepackage{algorithm}
\usepackage[noend]{algorithmic}
\usepackage{color}
\usepackage{subfig}
\usepackage[mathscr]{euscript}
\usepackage[noadjust]{cite}
\usepackage{multirow}
\usepackage{tabu,array,soul}
\usepackage{caption,xspace}
\usepackage{authblk}
\usepackage{framed}
\usepackage{environ}
\usepackage{tcolorbox}
\newcommand{\qedsymbol}{$\hfill\blacksquare$}
\hypersetup{
hidelinks,
colorlinks=true,
citecolor=[rgb]{0.121 0.47 0.705},
linkcolor=[rgb]{0.121 0.47 0.705},
urlcolor=[rgb]{0.121 0.47 0.705}
}

\usepackage{boxedminipage}
\newcommand{\old}[1]{{}}

\usepackage{thmtools}
\usepackage{thm-restate}

\usepackage{hyperref}

\usepackage{cleveref}

\begin{document}
\title{Algorithm and Hardness results on Liar's Dominating Set and $\boldmath{k}$-tuple Dominating Set}
\author{Sandip Banerjee\inst{1},
Sujoy Bhore\inst{2}}
\institute{Department of Computer Science, Hebrew University of Jerusalem, Israel\\\email{sandip.ndp@gmail.com},\and
Algorithms and Complexity Group, Technische Universit\"{a}t Wien, Austria  \\\email{sujoy@ac.tuwien.ac.at}
}
\maketitle
\vspace{-0.1 in}

\pagestyle{plain}
\pagenumbering{arabic}

\begin{abstract}

Given a graph $G=(V,E)$, the dominating set problem asks for a minimum subset of vertices $D\subseteq V$ 
such that every vertex $u\in V\setminus D$ is adjacent to at least one vertex $v\in D$. 
That is, the set $D$ satisfies the condition that $|N[v]\cap D|\geq 1$ for each $v\in V$,
where $N[v]$ is the closed neighborhood of $v$.
In this paper, we study two variants of the classical dominating set problem: $\boldmath{k}$-tuple dominating set ($k$-DS) problem and 
Liar's dominating set (LDS) problem, and obtain several algorithmic and hardness results. 

\vspace{.2cm}

On the algorithmic side, we present a constant factor ($\frac{11}{2}$)-approximation algorithm for the Liar's dominating set problem on unit disk graphs. 
Then, we obtain a PTAS for the $\boldmath{k}$-tuple dominating set problem on unit disk graphs.
On the hardness side, we show a $\Omega (n^2)$ bits lower bound for the space complexity of any (randomized) streaming algorithm for Liar's dominating set problem as well as for the $\boldmath{k}$-tuple dominating set problem. 
Furthermore, we prove that the Liar's dominating set problem on bipartite graphs is W[2]-hard.
\end{abstract}

\section{Introduction}
The dominating set problem is regarded as one of the fundamental problems in theoretical computer science
which finds its applications in various fields of science and engineering \cite{HSH98,GC98}.
A \textit{dominating set} of a graph $G=(V,E)$ is a subset $D$ of $V$
such that every vertex in $V\setminus D$ is adjacent to at least one vertex in $D$.
The \textit{domination number}, denoted as $\gamma(G)$, is the minimum cardinality of a dominating set of $G$.
Garey and Johnson \cite{garey2002computers} showed that deciding whether a given graph has
domination number at most some given integer $k$ is NP-complete.
For a vertex $v\in V$, the open neighborhood of the vertex $v$ denoted as $N_{G}(v)$ is defined as $N_G(v)=\{u|(u,v)\in E\}$
and the closed neighborhood of the vertex $v$ denoted as $N_G[v]$ is
defined as $N_G[v]=N_G(v)\cup \{v\}$.

\newpage

\paragraph{\textbf{\underline{$\boldmath{k}$-tuple Dominating Set ($\boldmath{k}$-DS)}}:} Fink and Jacobson~\cite{FJ85} generalized
the concept of dominating sets as follows.
\vspace{.1cm}
\begin{tcolorbox}[colback=blue!3!white]
	\noindent {\bf{\color{black} \textit{{\emph{$\boldmath{k}$-tuple Dominating Set (\textit{$\boldmath{k}$-DS}) Problem}}}}}\\
	{\bf Input:} A graph $G=(V,E)$ and a non-negative integer $k$.\\
	{\bf Goal:} Choose a minimum cardinality subset of vertices $D\subseteq V$ such that,
	for every vertex $v\in V$, $|N_G[v]\cap D|\geq k$.
\end{tcolorbox}

The \textit{$\boldmath{k}$-tuple domination number} $\gamma_{k}(G)$ is the minimum cardinality of a $\boldmath{k}$-DS of $G$.
A good survey on $\boldmath{k}$-DS can be found in \cite{CFHV12,LC03}.
Note that, $1$-tuple dominating set is the usual dominating set, and 
$2$-tuple and $3$-tuple dominating set are known as \textit{double dominating set}~\cite{FH00}
and \textit{triple dominating set}~\cite{RV07}, respectively. 
Further note that, for a graph $G=(V,E)$, $\gamma_{k}(G)=\infty$, if there exists no $\boldmath{k}$-DS of $G$.
Klasing et al. \cite{KL04} studied the $\boldmath{k}$-DS problem from hardness and approximation point of view.
They gave a $(\log |V|+1)$-approximation algorithm for the $k$-tuple domination problem in general graphs,
and showed that it cannot be approximated within a ratio of $(1-\epsilon)\log |V|$, for any $\epsilon>0$.

\vspace{-.5cm}

\paragraph{\textbf{\underline{Liar's Dominating Set (LDS)}:}}
Slater \cite{PJS09} introduced a variant of the dominating set problem called
Liar's dominating set problem. Given a graph $G=(V,E)$, in this problem the objective is to
choose minimum number of vertices $D\subseteq V$ such that each vertex $v\in V$ is double dominated and
for every two vertices $u,v\in V$ there are at least three vertices in $D$ from the union of their neighborhood set.
The LDS problem is an important theoretical model for the following real-world problem. 
Consider a large computer network where a virus (generated elsewhere in the Internet) can attack any of the processors in the network.
The network can be viewed as an unweighted graph. 
For each vertex $v\in G$, an anti-virus can: $(1)$ detect the virus at $v$ as well as in its closed neighborhood $N[v]$,
and $(2)$ find and report the vertex $u\in N[v]$ at which the virus is located.
Notice that, one can make network $G$ virus free by deploying the anti-virus at the vertices $v\in D$,
where $D$ is the minimum size \textit{dominating set}.
However, in certain situations the anti-viruses may fail.
Hence to make the system virus free it is likely to double-guard the nodes of the network, which is indeed the $2$-tuple DS.
But, if the anti-viruses detect the location of the viruses but fails to cure it properly 
due to some software error or corrupted circumstances. 
In this scenario, LDS serves the purpose. We define the problem formally below.

\vspace{.1cm}
\begin{tcolorbox}[colback=blue!3!white]
	\noindent {\bf{\color{black} \textit{{\emph{Liar's Dominating Set (\textit{LDS}) Problem}}}}}\\
	{\bf Input:} A graph $G=(V,E)$ and a non-negative integer $k$.\\
	{\bf Goal:} Choose a subset of vertices $L\subseteq V$ of minimum cardinality such that
\begin{itemize}
\item for every vertex $v\in V$, $|N_G[v]\cap L|\geq 2$,
\item for every pair of vertices $u, v\in V$ of distinct vertices $|(N_G[u]\cup N_G[v])\cap L|\ge 3$.
\end{itemize}
\end{tcolorbox}

\subsection{Our Results}
In this paper, we obtain seveal algorithmic and hardness results for LDS and $\boldmath{k}$-DS problems on various graph families. 
On the algorithmic side in section~\ref{algo}, we present a constant factor ($\frac{11}{2}$)-approximation algorithm for the Liar's dominating set (LDS) problem on unit disk graphs. Then, we obtain a PTAS for the $\boldmath{k}$-tuple dominating set ($\boldmath{k}$-DS) problem on unit disk graphs. On the hardness side in section~\ref{hardness}, we show a $\Omega (n^2)$ bits lower bound for the space complexity of any (randomized) streaming algorithm for Liar's dominating set problem as well as for the $\boldmath{k}$-tuple dominating set problem. Furthermore, we prove that the Liar's dominating set problem on bipartite graphs is W[2]-hard.

\section{Algorithmic Results}\label{algo}

\subsection{Approximation Algorithm for LDS on Unit Disk Graphs}\label{approx_const}
A unit disk graph (UDG) is an intersection graph  of a family of unit radius disks in the plane.
Formally, given a collection $\mathcal C=\{C_1,C_2,\ldots,C_n\}$ of $n$ unit disks in the plane,
a UDG is defined as a graph $G=(V,E)$, where each vertex $u\in V$ corresponds to a disk $C_i\in \mathcal C$
and there is an edge $(u,v)\in E$ between two vertices $u$ and $v$ if and only if their corresponding
disks $C_u$ and $C_v$ contain $v$ and $u$, respectively. Here, we consider the geometric variant of LDS known as
Euclidean Liar's dominating set problem, which is defined as follows:

\vspace{.2cm}

\begin{tcolorbox}[colback=blue!3!white]
	\noindent {\bf{\color{black} \textit{{\emph{Liar's Dominating Set on UDG (\textit{LDS-UDG}) Problem}}}}}\\
	{\bf Input:} A unit disk graph $G=(\mathcal P, E)$, where $\mathcal P$ is a set of $n$ disk centers.\\
	{\bf Output:} A minimum size subset $D\subseteq \mathcal P$ such that 
 \begin{itemize}
  \item for each point $p_i\in \mathcal P$, $|N[p_i]\cap D| \geq 2$.
  \item for each pair of points $p_i,p_j\in \mathcal P$, $|(N[p_i]\cup N[p_j])\cap D|\geq 3$.
 \end{itemize}
\end{tcolorbox}

Jallu et al. \cite{jallu2018liar} studied the LDS problem on unit disk graphs, and proved that this problem is NP-complete. Furthermore, given an unit disk graph $G=(V,E)$
and an $\epsilon>0$, they have designed a $(1+\epsilon)$-factor 
approximation algorithm to find an LDS in $G$ with running time 
$n^{O(c^2)}$, where $c=O(\frac{1}{\epsilon}\log \frac{1}{\epsilon})$.
In this section, we design a $\frac{11}{2}$-factor approximation algorithm that runs in sub-quadratic time.

For a point $p\in \mathcal P$, let $C(p)$ denote the disk centered at the point $p$.
For any two points $p,q\in \mathcal P$, if $q\in C(p)$, 
then we say that $q$ is a neighbor of $p$ (sometimes we say $q$ is covered by $p$) and vice versa.
Since for any Liar's dominating set $D$, $|(N[p_i]\cup N[p_j])\cap D|\geq 3$ ($\forall  p_i,p_j\in \mathcal P$) holds,
we assume that $|\mathcal P|\geq 3$ and for all the points $p\in \mathcal P$, $|N[p]|\geq 2$.
For a point $p_i\in \mathcal P$, let $p_i(x)$ and $p_i(y)$ denote the $x$ and $y$ coordinates of $p_i$, respectively.
Let $\text{Cov}_{\frac{1}{2}}(C(p_i))$, $\text{Cov}_1(C(p_i))$ ,$\text{Cov}_{\frac{3}{2}}(C(p_i))$ $\subseteq P$ denote the set of points
inside the circle centered at $p_i$ and of radius $\frac{1}{2}$, $1$ and $\frac{3}{2}$ unit, respectively.

The basic idea of our algorithm is as follows. 
Initially, we sort the points of $\mathcal P$ based on their $x$-coordinates.
Now, consider the leftmost point (say $p_i$).
We compute the sets $\text{Cov}_{\frac{1}{2}}(C(p_i))$, $\text{Cov}_1(C(p_i))$ and $\text{Cov}_{\frac{3}{2}}(C(p_i))$.
Next, we compute the set $Q=\text{Cov}_{\frac{3}{2}}(C(p_i))\setminus \text{Cov}_{\frac{1}{2}}(C(p_i))$.
Further, for each point $q_i\in Q$, we compute the set $S(q_i)=\text{Cov}_{1}(C(q_i))\cap \text{Cov}_{\frac{1}{2}}(C(p_i))$.
Finally, we compute the set $S=\bigcup S(q_i)$.
Moreover, our algorithm is divided into two cases.

\paragraph{Case~1 ($S\ne \emptyset$):}
For each point $q_i\in Q$ such that $S(q_i)\neq \emptyset$, we further distinguish between the following cases.

\begin{enumerate}
\item If $|S(q_i)|\geq 3$: we pick two arbitrary points from the set $S(q_i)$, and include them in the output set $D$ (see Figure~\ref{approx-case-1}(a)).

\item If $|S(q_i)|\leq 2$: in this case, we choose those two (or one) points (say $p_a,p_b\in S(q_i)$) in the output set $D$ 
(see Figure~\ref{approx-case-1}(b)).
\end{enumerate}

Once these points are selected, we remove the remaining points from 
$\text{Cov}_1(C(q_i))$ at this step from the set $Q$.
Notice that the points that lie in $\text{Cov}_1(C(q_i))$ are already $1$-dominated. 
Later, we can pick those points if required. 
However, observe that we may choose only $1$ point from $\text{Cov}_{\frac{1}{2}}(C(p_i))$ while we are in the case of $|S(q_i)|\leq 2$.
So we maintain a counter $t$ in Case~1. 
This counter keeps track of how many points we are picking from the set $S(q_i)$ in total, for each point $q_i\in Q$. 
If $t$ is at least $2$, we simply add $p_{i}$ to the output set and do not enter into Case~2. 
Otherwise, we proceed to Case~2. 

\vspace{2cm}
\begin{figure}[tbp]
 \centering
 \includegraphics[scale=.60]{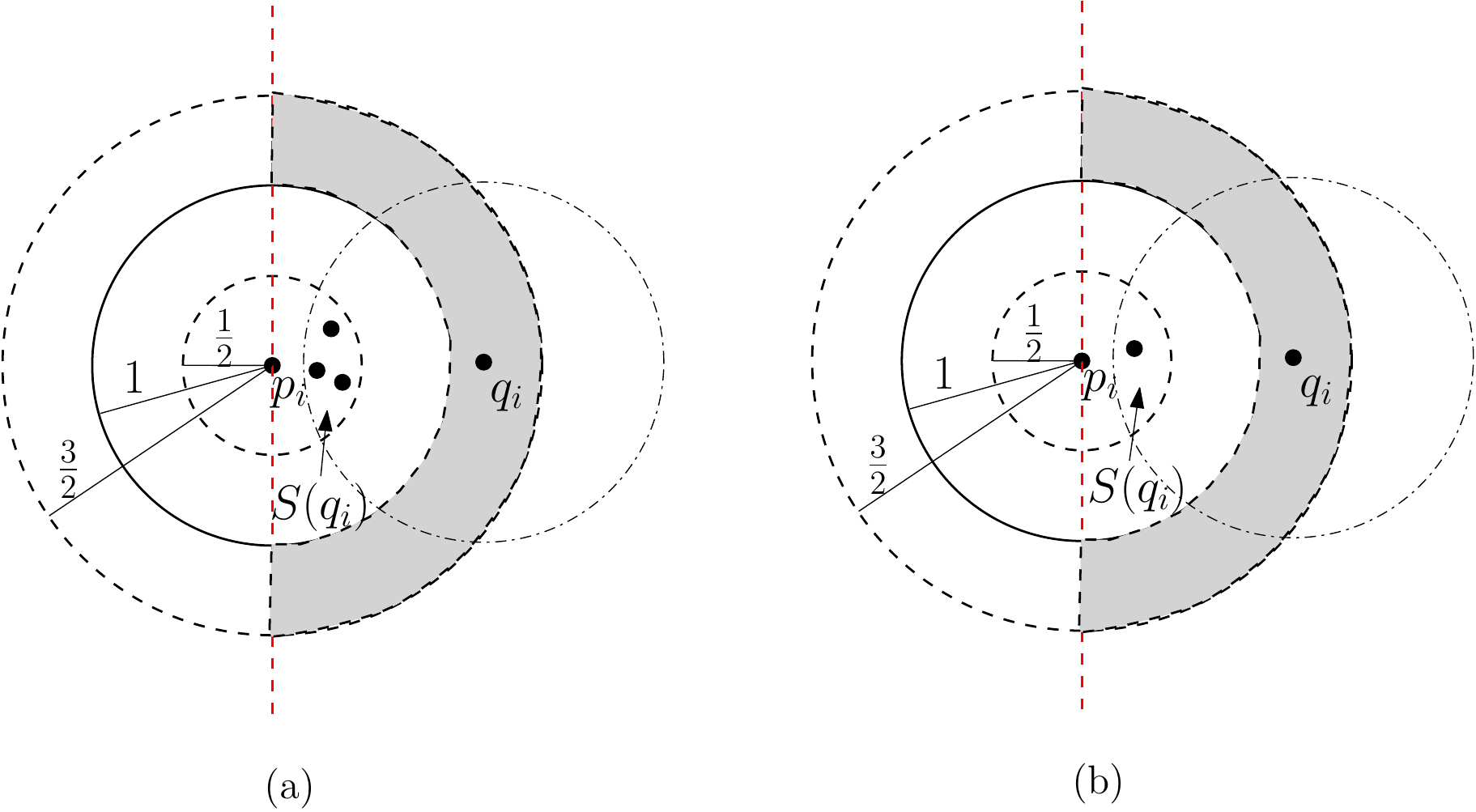}\vspace{-0.1in}
 \caption{An illustration of Case~1; (a) $|S(q_i)|\geq 3$,  (b) $|S(q_i)|\leq 2$.} \vspace{-0.1in}
 \label{approx-case-1}
\end{figure}

\newpage

\paragraph{Case~2 ($S=\emptyset$ or $t<2$):}
here, we further distinguish between the following cases.
\vspace{-4cm}
\begin{enumerate}
\item If $|\text{Cov}_{\frac{1}{2}}(C(p_i))|\geq 3$: then we choose $2$ points arbitrarily in the output set $D$ (see Figure~\ref{approx-case-2}(a)).
\vspace{-1cm}
\item If $|\text{Cov}_{\frac{1}{2}}(C(p_i))| = 2$: let $p_i,p_x\in \text{Cov}_{\frac{1}{2}}(C(p_i))$ be these points. 
We include both of them in the output set $D$.
This settles the first condition of LDS for them.
However, in order to  fulfill the second condition of LDS for $p_i$ and $p_x$, we must include at least one extra point here. 
First, we check the cardinality of $X=(\text{Cov}_1(C(p_i))\cap \text{Cov}_1(C(p_x)))\setminus \{p_i,p_x\}$.
If $|X|\neq \emptyset$, then we pick an arbitrary point $p_m$ from $X$, and include $p_m$ in $D$.
Otherwise, we include two points $p_l\in \text{Cov}_1(C(p_i))$ and $p_r\in \text{Cov}_1(C(p_x))$, and include them in $D$ (see Figure~\ref{approx-case-2}(b)).
Note that, in this case, we know that $p_l$ and $p_r$ exist due to the input constraint of an LDS problem. 
\vspace{-1cm}
\item If $|\text{Cov}_{\frac{1}{2}}(C(p_i))|=1$: then we check $\text{Cov}_1(C(p_i))$ and include two points
$p_m$ and $p_n$ arbitrarily from $\text{Cov}_1(C(p_i))\setminus \text{Cov}_{\frac{1}{2}}(C(p_i))$ (see Figure~\ref{approx-case-2}(c)).
\end{enumerate}
\vspace{-2cm}
This fulfills the criteria of LDS of points in $\text{Cov}_{\frac{1}{2}}(C(p_i))$.
Then we delete the remaining points of $\text{Cov}_{\frac{1}{2}}(C(p_i))$ from $\mathcal P$. Next we select the left-most point from the remaining  and repeat the same procedure until $\mathcal P$ is empty. The pseudocode is given in Algorithm \ref{alg:ours}.

\begin{figure}[tbp]
 \centering
 \includegraphics[scale=0.70]{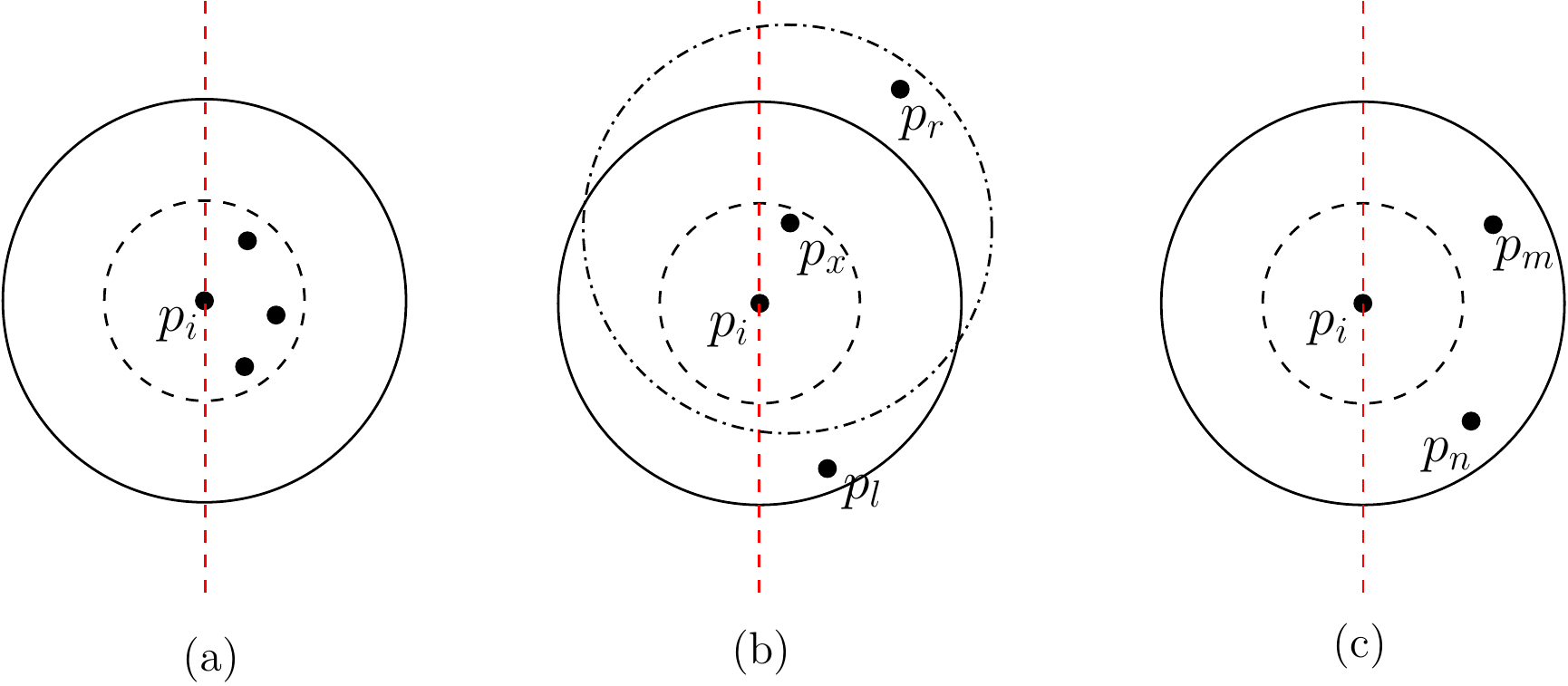}\vspace{-0.1in}
 \caption{An illustration of Case~2; (a) $|\text{Cov}_{\frac{1}{2}}(C(p_i))|\geq 3$,  (b) $|\text{Cov}_{\frac{1}{2}}(C(p_i))| = 2$, 
(c) $|\text{Cov}_{\frac{1}{2}}(C(p_i))|=1$.} \vspace{-0.1in}
 \label{approx-case-2}
 \end{figure}

\floatname{algorithm}{Procedure}
\begin{algorithm}[H]
\caption{Approximation Algorithm for LDS-UDG ($\mathcal P$)} \label{alg:ours}
Input: A set of points $\mathcal P$ in the plane.\\
Output: A subset of points $D\subseteq \mathcal P$, that is a liar's dominating set of the unit disk graph defined on the points of $\mathcal P$.
\begin{algorithmic}[1]

\STATE Sort the points of $\mathcal P$ based on their $x$-coordinates.\\ Let $\{p_1,p_2,\ldots, p_n\}$ be the sorted order.
\STATE Let $p_{left} \leftarrow p_1$
\STATE $t\leftarrow 0$ \hfill /*$t$ is a counter*/
\WHILE{$\mathcal P \ne \emptyset$}
\STATE $D= D\cup \{p_{left}\}$
\STATE Compute $\text{Cov}_{\frac{1}{2}}(C(p_{left}))$, $\text{Cov}_1(C(p_{left}))$, $\text{Cov}_{\frac{3}{2}}(C(p_{left}))$
\STATE Compute $Q=\text{Cov}_{\frac{3}{2}}(C(p_{left}))\setminus \text{Cov}_{\frac{1}{2}}(C(p_{left}))$
\FOR{$i \gets 1$ to $|Q|$}                    
        \STATE {Compute $S(q_i)=\text{Cov}_1(C(q_i))\cap \text{Cov}_{\frac{1}{2}}(C(p_{left}))$} \hfill /*$q_i\in Q$*/
\ENDFOR
\STATE Compute $S= \bigcup S(q_i)$
\IF{$S\ne \emptyset$}
\FOR{$i \gets 1$ to $|Q|$}                    
        \IF{$|S(q_i)|\ne \emptyset$}
				     \IF{$|S(q_i)| \geq 3$}
									\STATE $D=D\cup \{q_i,p_a,p_b\}$ \hfill  /*$p_a,p_b\in S(q_i)$ chosen arbitrarily*/
									\STATE $t\leftarrow t+2$
						 \ENDIF		
				\ELSIF{$S(q_i)=2$}
				     \STATE $D=D\cup \{q_i, p_a,p_b\}$  \hfill     /*$p_a,p_b\in S(q_i)$ are included in $D$*/
						 \STATE $t\leftarrow t+2$
				\ELSIF{$S(q_i)=1$}
						 \STATE $D=D\cup \{q_i, p_a\}$         \hfill       /*$p_a\in S(q_i)$ are included in $D$*/
						 \STATE $t\leftarrow t+1$
				\ENDIF
				\STATE $Q= Q\setminus \text{Cov}_{1}(C(q_{i}))$
\ENDFOR
\ENDIF
\IF{$t\geq 2$}
\STATE $D=D\cup \{p_{left}\}$ \\  \hfill /*Since already $2$ points are present from $\text{Cov}_{\frac{1}{2}}(C(p_{left}))$ in $D$*/
\STATE \textbf{Jump to} line~$35$
\ELSE 
\STATE $D= D$ $\cup$ \textbf{return} $SelectPoint(p_{left}, \mathcal P, \text{Cov}_{\frac{1}{2}}(C(p_{left})), \text{Cov}_{1}(C(p_{left})))$
\ENDIF
\STATE $\mathcal P=\mathcal P\setminus \text{Cov}_{\frac{1}{2}}(C(p_{left}))$
\STATE $p_{left}\leftarrow p_{next}$ \hfill /* $p_{next}$ is the next left-most point in $\mathcal P$*/
\ENDWHILE
\RETURN $D$.
\end{algorithmic}
\end{algorithm}
\clearpage

\floatname{algorithm}{Procedure}
\begin{algorithm}[ht]
\caption{$SelectPoint(p_{left}, \mathcal P, \text{Cov}_{\frac{1}{2}}(C(p_{left})), \text{Cov}_{1}(C(p_{left})))$} \label{alg:choosepoint}
Input: A point $p_{left}\in \mathcal P$, the sets $\text{Cov}_{\frac{1}{2}}(C(p_{left})), \text{Cov}_{1}(C(p_{left}))$. \\
Output: A subset of points $M\subset \mathcal P$.
\begin{algorithmic}[1]
\STATE Set $M\leftarrow \emptyset$
\STATE Set $X\leftarrow \emptyset$
\IF{$|\text{Cov}_{\frac{1}{2}}(C(p_{left}))|\geq 3$}
			\STATE $M= M \cup\{p_{left},p_i,p_j\}$ \hfill /*$p_i,p_j\in \text{Cov}_{\frac{1}{2}}(C(p_{left}))$ chosen arbitrarily.*/
			\ELSIF{$|\text{Cov}_{\frac{1}{2}}(C(p_{left}))|=2$}		
						\STATE $X=(\text{Cov}_1(C(p_{left}))\cap \text{Cov}_1(C(p_x)))\setminus \{p_{left},p_x\}$  \\
								\hfill /*$p_x$ is the other point in $\text{Cov}_{\frac{1}{2}}(C(p_{left}))$.*/
						\IF{$|X|\ne \emptyset$}
										\STATE $M=M\cup \{p_{left},p_x,p_m\}$ \hfill /*$p_m\in X$ is chosen arbitrarily.*/
						\ELSE 
										\STATE $M=M\cup \{p_{left},p_x,p_m,p_n\}$  \\ 
										\hfill /*$p_m\in \text{Cov}_1(C(p_{left}))$ and  $p_n\in \text{Cov}_1(C(p_x))$ is chosen arbitarily.*/
						\ENDIF		
			\ELSE 
						\STATE $M=M\cup \{p_{left},p_m,p_n\}$  \\
									\hfill 	/*$p_m, p_n\in \text{Cov}_1(C(p_{left})) \setminus \text{Cov}_{\frac{1}{2}}(C(p_{left}))$ chosen arbitrarily.*/
\ENDIF	
\RETURN $M$
\end{algorithmic}
\end{algorithm}

\begin{lemma}\label{correctness}
The set $D$ obtained from Algorithm~\ref{alg:ours}, is a liar's dominating set of the unit disk graph defined on the points of $\mathcal P$. \label{lem:lds-approx-correctness}
\end{lemma}

\begin{proof}
Algorithm~\ref{alg:ours} primarily relies on two cases.
We begin with the leftmost point (say $p_i$).
We first compute the sets $\text{Cov}_{\frac{1}{2}}(C(p_i))$, $\text{Cov}_{1}(C(p_i))$ and $\text{Cov}_{\frac{3}{2}}(C(p_i))$.
Then, we consider the points that lie in $Q= (\text{Cov}_{\frac{3}{2}}(C(p_i))\setminus \text{Cov}_{\frac{1}{2}}(C(p_i)))$.
If the disks centered at these points of radius~$1$ share points with $\text{Cov}_{\frac{1}{2}}(C(p_i))$,
we choose points from their intersection to $D$.
This choice is based on the number of points available in the intersection. 
For every point $p_k\in Q$,
if $|S(p_k)|\geq 3$ ($S(p_k)=\text{Cov}_1(C(p_k))\cap \text{Cov}_{\frac{1}{2}}(C(p_i))$), we arbitrarily choose $2$ points in $D$.
Otherwise, we choose one or two points in $D$, based on the cardinality of $S(p_k)$.
Once these points are selected, we remove the remaining points of $\text{Cov}_1(C(p_k))$ at this step from the set $Q$.
Notice that the points lying in $\text{Cov}_1(C(q_i))$ are already $1$-dominated. 
Later, we can pick those points if required.  
Clearly, after this process any point that lies in $\text{Cov}_{\frac{3}{2}}(C(p_i))$ and
shares points with $\text{Cov}_{\frac{1}{2}}(C(p_i))$ is at least $1$-dominated or $2$-dominated,
Notice that for any point that lies outside $\text{Cov}_{\frac{3}{2}}(C(p_i))$, its corresponding disk does not share any region with $\text{Cov}_{\frac{1}{2}}(C(p_i))$.
Later, for any points in $\text{Cov}_{\frac{3}{2}}(C(p_i))$ if we have to choose points from $\text{Cov}_{\frac{1}{2}}(C(p_i))$ to fulfill the 
criteria of LDS, we already took them in the output set at this stage. 
Meanwhile, we maintain a counter $t$ that keeps track of how many points we are picking from $\text{Cov}_{\frac{1}{2}}(C(p_i))$, for each point $p_k\in Q$. 
If $t$ is at least $2$, we know $|\text{Cov}_{\frac{1}{2}}(C(p_i))\cap D| \geq 2$.
We simply add $p_{i}$ to the output set and do not enter into Case~2.
Otherwise, we proceed to Case~2. 
Besides, if there is no point $p_x\in (\text{Cov}_{\frac{3}{2}}(C(p_i))\setminus \text{Cov}_{\frac{1}{2}}(C(p_i)))$,
such that the disk centered at $p_x$ shares points with $\text{Cov}_{\frac{1}{2}}(C(p_i))$, we are in Case~2. 

Now we show the correctness of Case~2. 
In this case, if $|\text{Cov}_{\frac{1}{2}}(C(p_i))|\geq 3$, then we include three points arbitrarily including $p_i$ (say $p_i, p_j, p_k$ in $D$).
Note that, the distance between any two points $p_l,p_r\in \text{Cov}_{\frac{1}{2}}(C(p_i))$ is at most $1$.
Since we are including $3$ points from $\text{Cov}_{\frac{1}{2}}(C(p_i))$ in $D$, it fulfills the criteria 
of LDS of the points contained inside the unit disk centered at $p_i$.
Next, consider the case if $|\text{Cov}_{\frac{1}{2}}(C(p_i))|=2$. 
Let $p_x$ be the other point in $\text{Cov}_{\frac{1}{2}}(C(p_i))$. 
We compute $X=(\text{Cov}_1(C(p_i))\cap \text{Cov}_1(C(p_x)))\setminus \{p_i,p_x\}$.
Note that, due to the definitions of the problem, we know
$|\text{Cov}_1(C(p_i))|\geq 3$ and $|\text{Cov}_1(C(p_x))|\geq 3$.
When $|X|\neq \emptyset$ we include an additional point arbitrarily (say $p_m\in X$)
in $D$ to fulfill the second condition of the LDS of the points $p_i, p_x \in \text{Cov}_{\frac{1}{2}}(C(p_1))$.
While $|X|=\emptyset$, we include $2$ points $p_m$ and $p_n$
(both the points are chosen arbitrarily) from $\text{Cov}_1(C(p_i))$ and $\text{Cov}_1(C(p_x))$,
respectively. Notice that this makes the points $p_i$ and $p_x$ individually $2$ dominated. 
Now, for any pair of points from $\text{Cov}_1(C(p_i))$, at least $3$ points are chosen (that is, $p_i,p_x,p_m$). 
Finally, when $|\text{Cov}_{\frac{1}{2}}(C(p_i))|=1$, we include
two points $p_m, p_n \in \text{Cov}_1(C(p_i))$ arbitrarily in $D$
to fulfill the criteria of LDS for $p_i$. 
Note that the existence of two points $p_m$ and $p_n$ in the set $\text{Cov}_1(C(p_i))$ follows from the problem constraint.

Thus, at every iteration, we made sure that the considered disk and the points inside that disk, fulfills the criteria of LDS.
Furthermore, due to Case~1, we have taken sufficient points from the considered disk for the disks that intersect with it. This completes the proof. 
\qed 
\end{proof}

\begin{lemma}\label{approx-factor} 
Algorithm~\ref{alg:ours} outputs a liar's dominating set $D\subseteq \mathcal P$ of the unit disk graph defined on the points of $\mathcal P$ with approximation ratio ${\frac{11}{2}}$.
\end{lemma}

\begin{proof}
For each point $p_i\in \mathcal P$, we consider the set $\text{Cov}_{\frac{1}{2}}(C(p_i))$ and choose points from this set.
Consider the disk $C(p_i)$ of radius $\frac{3}{2}$.
Notice that the points that lie outside the set $\text{Cov}_{\frac{3}{2}}(C(p_i))$, their corresponding disks of radius~$1$, do not share any points with the set
 $\text{Cov}_{\frac{1}{2}}(C(p_i))$.
Thus, we only need to consider the points inside $\text{Cov}_{\frac{3}{2}}(C(p_i))$.
The half-perimeter of the of the disk $C(p_i)$ of radius $\frac{3}{2}$ is $(\pi\cdot {\frac{3}{2}}) \approx 4.71$ 
(since we are only considering the points which are to the right of $p_i$). 
Thus, we can pack at most $5$ disks of radius~$1$ inside $C(p_i)$ such that they mutually do not contain the center of other disks 
(see Figure~\ref{fig:approx}).
We consider the points that lie in $(\text{Cov}_{\frac{3}{2}}(C(p_i))\setminus \text{Cov}_{\frac{1}{2}}(C(p_i)))$.
If the disks centered at these points of radius~$1$ share points with $\text{Cov}_{\frac{1}{2}}(C(p_i))$,
we choose points from their intersection to $D$.
Now, for each point $q_i\in (\text{Cov}_{\frac{3}{2}}(C(p_i))\setminus \text{Cov}_{\frac{1}{2}}(C(p_i)))$, we choose at most two points 
from $\text{Cov}_{\frac{1}{2}}(C(p_i)) \cap \text{Cov}_{1}(C(q_i))$ (see Procedure~\ref{alg:ours}). 
Once these points are selected, we remove the remaining points at this step from the set $Q$ (recall that $Q=(\text{Cov}_{\frac{3}{2}}(C(p_i))\setminus \text{Cov}_{\frac{1}{2}}(C(p_i)))$).
The points that lie inside this disk, is already $1$-dominated.
Later, we can pick the point itself, if required. 
Besides we also pick $p_i$. 
In this case, we pick at most $(5\times 2)+1 = 11$.

\begin{figure}[h]
 \centering
 \includegraphics[scale=0.5]{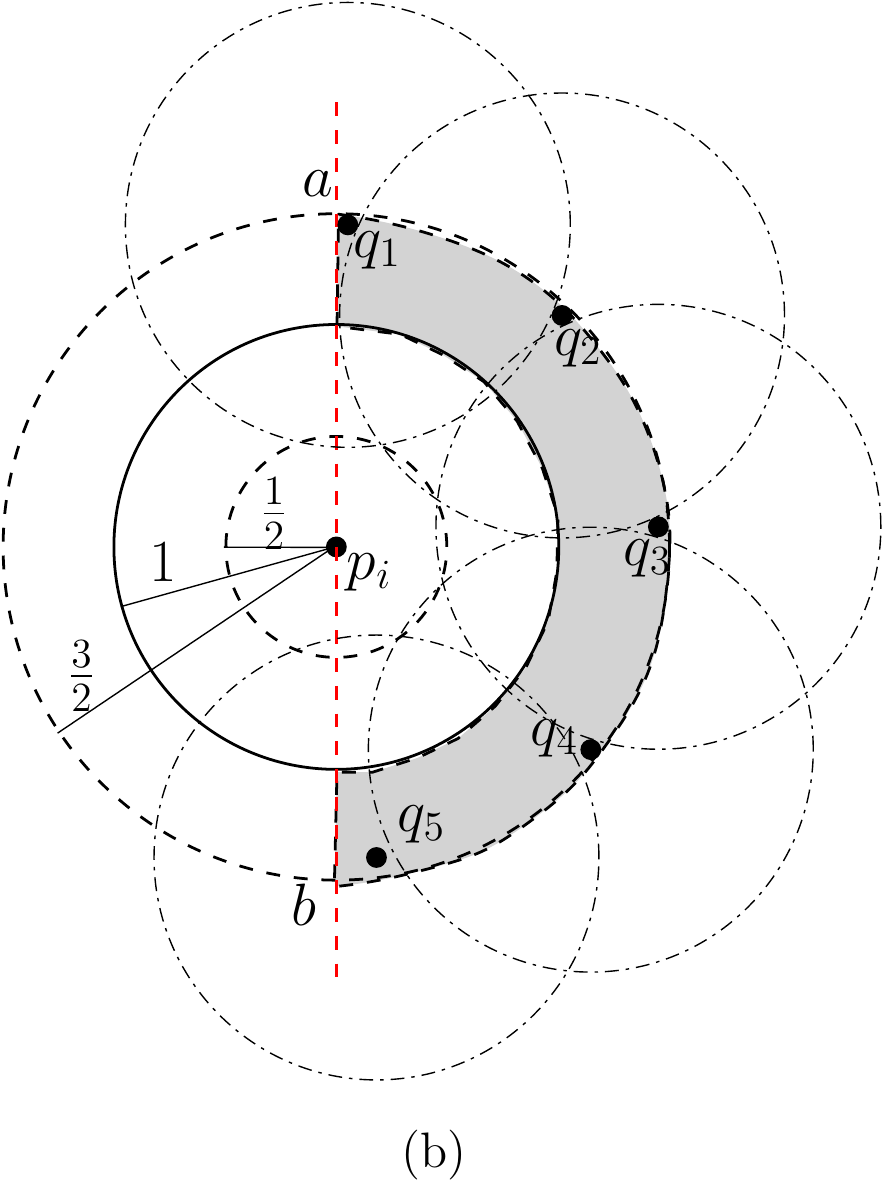}\vspace{-0.1in}
 \caption{An illustration of Lemma~\ref{approx-factor}; at most $5$ disks of radius~$1$ can be packed such that they mutually do not contain the center of other disks.} \vspace{-0.1in}
 \label{fig:approx}
 \end{figure}

Now, if there is no point in $(\text{Cov}_{\frac{3}{2}}(C(p_i))\setminus \text{Cov}_{\frac{1}{2}}(C(p_i)))$, whose corresponding disk that shares a region $\text{Cov}_{\frac{1}{2}}(C(p_i))$,
or the number of points picked from $\text{Cov}_{\frac{1}{2}}(C(p_i))$ is less than $2$, we enter into Case~2 (see Procedure~\ref{alg:choosepoint}).
Here, we include at most three points including the point $p_i$. However, we enter into this case only when sufficient points have not been picked in Case~1. 
Thus, in this case we have picked fewer than $11$ points from $\text{Cov}_{1}(C(p_i))$ (precisely, we pick $4$ points). 

Now, any optimal solution contains at least two points from $\text{Cov}_{1}(C(p_i))$, for each point $p_i\in \mathcal P$.
Since we have to fulfill the first constraint of the LDS that each disk is individually $2$ dominated. 
Hence, we get the approximation factor ${\frac{11}{2}}$. 
\qed
\end{proof}

Algorithm~\ref{alg:ours} clearly runs in polynomial time (to be precise in sub-quadratic time). Thus, from Lemma~\ref{correctness} and Lemma~\ref{approx-factor} we conclude the following theorem. 

\begin{theorem}
Algorithm~\ref{alg:ours} computes a liar's dominating set $D\subseteq \mathcal P$ of the unit disk graph defined on the points of $\mathcal P$ in sub-quadratic time
with approximation ratio ${\frac{11}{2}}$.
\end{theorem}

\subsection{PTAS for $\boldmath{k}$-DS on Unit Disk Graphs}\label{ptas_unitdisk}
In this section we give a PTAS for the $k$-tuple dominating set on unit disk graphs with a similar approach used by Nieberg and Hurink~\cite{NH05}. It might be possible to design a PTAS by using local search or shifting strategy for the $k$-tuple dominating set problem on unit disk graphs. However, the complexity of these algorithms would be dependent on $n$ and $k$. 
Thus we use the approach of Nieberg and Hurink~\cite{NH05}, that gurantees a better running time. 

Let $G=(V,E)$ be an unit disk graph in the plane. 
For a vertex $v\in V$, let $N^{r}(v)$ and $N^{r}[v]=N^{r}(v) \cup \{v\}$ be the
$r$-th neighborhood and $r$-th closed neighborhood of $v$, respectively.
For any two vertices $u,v\in V$, let $\delta(u,v)$ be the distance
between $u$ and $v$ in $G$, that is the number of edges of a shortest path between $u$ and $v$ in $G$.
Let $D_{k}(V)$ be the minimum $k$-tuple dominating set of $G$.
For a subset $W\subseteq V$, let $D_{k}(W)$ be the minimum $k$-tuple dominating set of
the induced subgraph on~$W$.
We prove the following theorem. 

\begin{theorem}
There exists a PTAS for the $k$-tuple dominating set problem on unit disk graphs.
\end{theorem}

\begin{proof}
The $2$-separated collection of subsets is defined as follows:
Given a graph $G=(V,E)$, let $S=\{S_1,\ldots,S_m\}$ be a collection of
subsets of vertices $S_i\subset V$, for $i=1,\ldots,m$, such that
for any two vertices $u\in S_i$ and $v\in S_j$ with $i\ne j$, $\delta(u,v)>2$.
In the following lemma we prove that the sum of the cardinalities of the minimum 
$k$-tuple dominating sets $D_{k}(S_i)$ 
for the subsets $S_i\in S$ of a $2$-separated collection is 
a lower bound on the cardinality of $D_{k}(V)$.

\begin{lemma}\label{lem:ptas}
Given a graph $G=(V,E)$, let $S=\{S_1,\ldots,S_m\}$ be a 2-separated
collection of subsets then, $|D_k(V)|\ge \sum_{i=1}^{m} |D_k(S_i)|$.
\label{lem:ptas}
\end{lemma}

\begin{proof}
Consider a subset $S_i\in S$, and its neighborhood $N[S_i]$.
We know from the properties of $2$-separated collection of subsets that,
any two subsets $S_i,S_j\in S$ are disjoint for each $i\ne j$.
Any vertex outside $N[S_i]$ has at least distance more than one to any vertex in $S_i$.
Hence, $D_k(V)\cap N[S_i]$ has to be a $k$-tuple domination of the vertices of $S_i$, where
$D_k(V)$ is the $k$-tuple domination of $G$.
Otherwise the vertices of $S_i$ is not $k$-tuple dominated,
since no vertex outside $N[S_i]$ can dominate a vertex of $S_i$.
On the other hand, we know $D_{k}(S_i)\subset N[S_i]$.
This implies $D_{k}(S_i) \leq D_{k}(V)\cap N[S_i]$.
Hence, $\sum_{i=1}^{m} |D_{k}(S_i)| \leq \sum_{i=1}^{m}|D_{k}(V)\cap N[S_i]| \leq |D_{k}(V)|$.~\qedsymbol
\end{proof}

From Lemma~\ref{lem:ptas}, we get the lower bound of the minimum
$k$-tuple dominating set of $G$.
If we can enlarge each of the subset $S_i$ to a subset $T_i$
such that the $k$-tuple dominating set of $S_i$ (that is $D_{k}(S_i)$)
is locally bounded to the $k$-tuple dominating set of $T_i$ (that is $D_{k}(T_i)$),
then by taking the union of them we get an approximation of the $k$-tuple dominating set of $G$.
For each subset $S_i$, let there is a subset $T_i$ (where $S_i\subset T_i$),
and let there exists a bound $(1+\epsilon)$ ($0<\epsilon<1$) such that
$|D_{k}(T_i)|\leq (1+\epsilon)\cdot |D_{k}(S_i)|$.
Then, if we take the union of the $k$-tuple dominating sets of all $T_i$,
this is a $(1+\epsilon)$-approximation of the $k$-tuple dominating sets of the 
union of subsets $S_i$ (for $i=1,\ldots,m$).

Now, we describe the algorithm. 
We begin with an arbitrary vertex $v\in V$, and compute
the $k$-tuple dominating sets of minimum cardinality for these neighborhoods until $D_{k}(N^{r+2}[v]) > \rho \cdot D_{k}(N^{r}[v])$ 
(for a constant $\rho$).
Then, this process is used iteratively on the remaining graph induced by
$V_{i+1} = V_{i}\setminus N^{\hat{r_i}+2}[v_i]$ (where $\hat{r_i}$ is the first point at $i$th iteration when the condition is violated).
Let $\ell$ be the total number of iterations, where $\ell < n$.
Let $\{N_1,\ldots,N_{\ell}\}$ be the neighborhoods achieved from this process.
The following lemma shows that a $k$-tuple dominating set for the entire graph $G$
is given by the union of the sets $D_{k}(N_i)$.

\begin{lemma}\label{lem:neighbor}
Let $\{N_1,\ldots,N_{\ell}\}$ be the set of neighborhoods created by the above algorithm. 
The union $\bigcup^{\ell}_{i=1} (D_{k}(N_i))$ forms a $k$-tuple 
dominating set of $G$.
\end{lemma}

\begin{proof}
Consider the set $V_{i+1}=V_{i}\setminus N_{i}$, and we know $N_{i} \subset V_{i}$. 
Thus, $V_{i+1}=V_i \cup N_{i}$.
The algorithm stops while $V_{\ell+1}=\emptyset$, which means $V_{\ell}=N_{\ell}$.
Besides, $\bigcup^{\ell}_{i=1} (N_i)=V$.
Thus, if we compute the $k$-tuple dominating set $D_{k}(N_i)$ of each $N_i$, their union clearly is the 
$k$-tuple dominating set of the entire graph.
~\qedsymbol
\end{proof}

These subsets $N^{\hat{r_i}}[v_i]$, for $i=1,\ldots, \ell$, created by the algorithm
form a $2$-separated collection $\{N^{\hat{r_1}}[v_1],\ldots, N^{\hat{r_{\ell}}}[v_{\ell}]\}$ in $G$.
Consider any two neighborhoods $N_i,N_{i+1}$.
We have computed $N_{i+1}$ on graph induced by $V\setminus N_i$ (where $N_i = N^{\hat{r_i}+2}[v_i]$).
So, for any two vertices $u\in N_i$ and $v\in N_{i+1}$, the distance is at least $2$.
Thus we have the following corollary. 

\begin{corollary}
The algorithm returns a $k$-tuple dominating set $\bigcup^{\ell}_{i=1} (D_{k}(N_i))$ of cardinality no more
than $(1+\epsilon)$ the size of a dominating set $D_{k}(V)$.
\end{corollary}

It needs to be shown that this algorithm has a polynomial running time.
The number of iterations $\ell$ is clearly bounded by $|V|=n$.
It is important to show that for each iteration
we can compute the minimum $k$-tuple dominating set $D_{k}(N^{r}[v])$ in polynomial time
for $r$ being constant or polynomially bounded.

Consider the $r$-th neighborhood of a vertex $v$, $N^{r}[v]$.
Let $I^{r}$ be the maximal independent set of the graph induced by $N^{r}[v]$.
From \cite{NH05}, we have $I^{r}\leq (2r+1)^2 = O(r^2)$.
The cardinality of a minimum dominating set in $N^{r}[v]$
is bounded from above by the cardinality of a maximal independent set in $N^{r}[v]$.
Hence, $|D(N^{r}[v]| \leq (2r+1)^2 = O(r^2)$.
Now, we prove the following lemma.

\begin{lemma}
$|D_{k}(N^r[v])| \leq O(k^2\cdot r^2)$.
\label{fixed_kr}
\end{lemma}

\begin{proof}
For a vertex $v\in V$, let $N^{r}[v]$ be the $r$-th closed neighborhood of $v$.
Let $I^{r}_1$ be the first maximal independent set of $N^{r}[v]$.
We know $|D(N^{r}[v])| \leq |I^{r}_1| \leq (2r+1)^2$.

Now, we take the next maximal independent set $I^{r}_2$ from $N^{r}[v] \setminus I^{r}_1$, and take the union of them $(I^{r}_1 \cup I^{r}_2)$.
Notice that every vertex $v\in (N^{r}[v] \setminus (I^{r}_1\cup I^{r}_2))$ has $2$ neighbors in $(I^{r}_1\cup I^{r}_2)$,
so they can be $2$-tuple dominated by choosing vertices from $(I^{r}_1\cup I^{r}_2)$.  
Also, every vertex $v\in I^{r}_2$ can be $2$-tuple dominated by choosing vertices from $(I^{r}_1\cup I^{r}_2)$, 
since $v$ itself can be one and the other one can be picked from $I^{r}_1$. 
Additionally, for each vertex in $u \in I^{r}_1$, we take a vertex $z$ from the neighborhood of $u$ in $(N^{r}[v] \setminus (I^{r}_1\cup I^{r}_2))$.
Let $W(I^{r}_1)$ be the union of these vertices. 
Now, every vertex $v\in N^{r}[v]$ can be $2$-tuple dominated by choosing vertices from $(I^{r}_1 \cup I^{r}_2 \cup W(I^{r}_1))$. 
So, $|D_{2}(N^{r}[v])| \leq |(I^{r}_1 \cup I^{r}_2 \cup W(I^{r}_1))|$.
$|(I^{r}_1 \cup I^{r}_2 \cup W(I^{r}_1))| \leq 3\cdot (2r+1)^2$. Hence, $|D_{2}(N^{r}[v])| \leq 3\cdot (2r+1)^2$.
We continue this process $k$ times.

After $k$ steps, we get the union of the maximal independent sets  $A=\{I^{r}_1 \cup \ldots \cup I^{r}_{k}\}$.
Additionally, we get the unions of $B=\{W(I^{r}_1) \cup W(I^{r}_1\cup I^{r}_2) \cup \ldots \cup W(I^{r}_1\cup \ldots \cup I^{r}_{k-1})\}$.
Notice that every vertex $v\in N^{r}[v]$ can be $k$-tuple dominated by choosing vertices from $(A\cup B)$.
The cardinality of $(A\cup B)$ is at most $(2r+1)^2 \cdot (1+3+\ldots+(2k-1))$, which is $(2r+1)^2 \cdot k^2$.
We also know $|D_{k}(N^{r}[v])|$ is upper bounded by $(A\cup B)$. 
Thus, $|D_{k}(N^{r}[v])| \leq (2r+1)^2 \cdot k^2 \leq O(k^2\cdot r^2)$.
~\qedsymbol
\end{proof}

Nieberg and Hurink~\cite{NH05} showed that for a unit disk graph, there exists a bound on $\hat{r_1}$
(the first value of $r$ that violates the property $D(N^{r+2}[v]) > \rho \cdot D(N^{r}[v]))$.
This bound depends on the approximation $\rho$ not on the size of the of the unit disk graph $G=(V,E)$ given as input.
Precisely, they have proved that there exists a constant $c = c(\rho)$ such that $\hat{r_1} \leq c$,
that is, the largest neighborhood to be considered during the iteration of the algorithm is bounded by a constant.
Putting everything together, we conclude the proof.
\qed
\end{proof}

\section{Hardness Results}\label{hardness}

\subsection{Streaming Lower Bound for LDS}
In this section, we study the LDS problem in the streaming model. In the streaming model, the vertex set is fixed, and the edges appear one-by-one in sequential order. We are required to decide whether we remember this edge or forget at this particular time step. We show that in order to solve the LDS problem, one needs to essentially store the entire graph. We prove the following theorem by using similar approach used in \cite{chitnis2014parameterized}.

\begin{theorem}
Any randomized (by randomized algorithm we mean that the algorithm should succeed
with probability $\geq \frac{2}{3}$) streaming algorithm for LDS problem on $n$-vertex graphs 
requires $\Omega(n^2)$ space.
\label{thm:lower-bound-lds}
\end{theorem}

\begin{proof}
We reduce from the \textsc{Index} problem in communication complexity. 
In the \textsc{Index} problem, Alice has a string $X\in \{0,1\}^N$ given by
$x_{1}\ldots x_N$ and Bob has an index $i\in [N]$. Now, Bob wants to find $x_i$, i.e., the $i^{th}$ bit of $X$.
There is a lower bound of $\Omega(N)$ bits in the one-way randomized communication model for Bob to compute $x_i$ (for any $i\in [N]$)~\cite{nisan}.
Let us consider an instance of the \textsc{Index} problem where $N$ is a perfect
square, and let $r=\sqrt{N}$. Let us fix any bijection from $[N]\rightarrow [r] \times [r]$. We interpret the bit string as an adjacency matrix of a
bipartite graph with $r$ vertices on each side, where $V, W$ are two 
partites. 

Now, we construct an instance $G_X$ of the LDS. Assume that, Alice has an algorithm that solves the LDS problem by using $f(r)$ bits.
Firstly, we insert the edges corresponding to the edge interpretation of
$X$ between nodes $v_i$ and $w_j$: for each $i,j \in [k]$, 
Alice adds the edge $(v_i, w_j)$ if the corresponding entry in $X$ is $1$, and sends the memory contents of her algorithm to Bob by using $f(r)$ bits. Now, Bob has the index $i \in [N]$ that he interprets as $(I,J)$, under the same bijection $\phi:[N]\rightarrow [r] \times [r]$.
He receives the memory contents of the algorithm, and does the following --- (1) adds two vertices $a$ and $b$, and an edge $a-b$,
(2) adds an edge from each vertex of $V\setminus v_I$ to $a$,
(3) adds an edge from each vertex of $W\setminus w_J$ to $a$,
(4) adds five vertices $\{u,y,u',y',z\}$ and edges $u-u', y-y', u-z$ and $y-z$, (5) adds an edge from each vertex of $V\cup W\cup \{a,b\}$ to each vertex from $\{u,y\}$.

Let $D$ be a minimum LDS of $G_X$. Notice that, firstly, $D$ needs to be a $2$-tuple DS of $G_X$. 
Since $u'$ has only $2$ neighbors in $G_X$, it follows that $\{u,u'\}\subseteq D$. 
Similarly, $\{y,y'\}\subseteq D$. 
Note, $z$ also has only two neighbors in $G_X$.
Therefore, $N[z]\cup N[b] = \{u,y,a,b\}$. 
Since we should have $|(N[z]\cup N[b])\cap D| \geq 3$,
at least one of $a/b$ must be in $D$. 
Since $N[b]\subseteq N[a]$, w.l.o.g., we can assume that $a\in D$. 
Therefore, we conclude that $\{u,u',y,y',a\}\subseteq D$.

Next, we show that finding the minimum value of a LDS of $G_X$
allows us to solve the corresponding instance $X$ of \textsc{Index}.

\begin{lemma}
$x_i = 1$ implies that the minimum size of a liar's dominating set of $G_X$ is $6$.
\label{thm:lds-1-implies-6}
\end{lemma}

\begin{proof}
 Suppose that $x_i = 1$, i.e., $v_{I}-w_J$ is an edge in $G_X$. We now claim that $D:=\{u,u',y,y',a\}\cup v_I$ is a LDS of $G_X$.

First, we check that $D$ is a $2$-tuple DS of $G_X$;
\begin{itemize}
  \item For each vertex in $\lambda\in G_X\setminus \{u,u',y,y',z\}$ we have $(N[\lambda]\cap D)\supseteq \{u,y\}$
  \item $(N[u]\cap D)\supseteq \{u,u'\}$
  \item $(N[y]\cap D)\supseteq \{y,y'\}$
  \item $(N[z]\cap D)= \{u,y\}$
  \item $(N[u']\cap D)= \{u,u'\}$
    \item $(N[y']\cap D)= \{y,y'\}$
\end{itemize}

Now we check the second condition. 
Let $T=G_X\setminus \{u,u',y,y',z\}$, and $T' = G_X\setminus T$
\begin{itemize}
\item For each $\lambda\in T\setminus \{v_I, w_J\}$ and each $\delta\in T'$ we have $(N[\lambda]\cup N[\delta])\cap D \supseteq \{a,u,y\} $
\item For each $\delta\in T'$ we have $(N[v_I]\cup N[\delta])\cap D \supseteq \{v_I,u,y\} $
\item For each $\delta\in T'$ we have $(N[w_J]\cup N[\delta])\cap D \supseteq \{v_I,u,y\} $
\item Now we consider pairs where both vertices are from $T'$. By symmetry, we only have to consider following choices
        \begin{itemize}
        \item $(N[u']\cup N[u])\cap D = \{u,u',y\}$
        \item $(N[u']\cup N[z])\cap D = \{u,u',y\}$
        \item $(N[u']\cup N[y])\cap D = \{u,u',y,y'\}$
        \item $(N[u']\cup N[y'])\cap D = \{u,u',y,y'\}$
        \end{itemize}
\item Now we consider pairs where both vertices are from $T$. By symmetry, we only have to consider following choices
        \begin{itemize}
        \item For each $\lambda\in V\setminus v_I \cup W\setminus w_J$ we have $(N[\lambda]\cup N[b])\cap D \supseteq \{u,y,a\}$ and  $(N[\lambda]\cup N[a])\cap D \supseteq \{u,y,a\}$
        \item $(N[v_I]\cup N[b])\cap D \supseteq \{u,y,a\}$
        \item $(N[v_I]\cup N[a])\cap D \supseteq \{u,y,a\}$
        \item $(N[w_J]\cup N[b])\cap D \supseteq \{u,y,a\}$
        \item $(N[w_J]\cup N[a])\cap D \supseteq \{u,y,a\}$
        \item $(N[w_J]\cup N[v_I])\cap D \supseteq \{v_I,y,a\}$
        \item For each $\gamma\in V\setminus v_I$ we have $(N[w_J]\cup N[\gamma])\cap D \supseteq \{v_I,y,a,u\}$
        \item For each $\gamma\in W\setminus w_J$ we have $(N[v_I]\cup N[\gamma])\cap D \supseteq \{v_I,y,a,u\}$
        \item For each $\gamma\in V\setminus v_I$ and $\gamma'\in W\setminus w_J$ we have $(N[\gamma]\cup N[\gamma'])\cap D \supseteq \{y,a,u\}$
        \end{itemize}
\end{itemize}
Hence, it follows that $D$ is indeed a LDS of $G_X$ of size 6.
~\qedsymbol
\end{proof}

\begin{lemma}
$x_i = 0$ implies that the minimum size of a LDS of $G_X$ is $\geq 7$.
\label{lemma_lds7}
\end{lemma}
\begin{proof}
 Now suppose that $x_i = 0$, i.e., $v_{I}$ and $w_J$ do not have an edge between them in $G_X$.
Let $D'$ be a minimum LDS of $G_X$. We have already seen above that $\{u,u',y,y',a\}\subseteq D$.
Consider the pair $(v_I, z)$. Currently, we have that $(N[v_I]\cup N[z])\cap \{u,u',y,y',a\} = \{u,y\}$.
Hence, $D'$ must contain a vertex, say $\mu\in N[v_I]\setminus \{u,y\}$. Consider the pair $(w_J, z)$.
Currently, we have that $(N[w_J]\cup N[z])\cap \{u,u',y,y',a\} = \{u,y\}$.
Hence, $D'$ must contain a vertex, say $\mu'\in N[w_J]\setminus \{u,y\}$.
Since $v_I$ and $w_J$ do not form an edge, we have that $\mu\neq \mu'$. Hence, $|D'|\geq  7$.
~\qedsymbol
\end{proof}

Thus, by checking whether the value of a minimum LDS on the instance
$G_X$ is $6$ or $7$, Bob can determine the index
$x_i$.
The total communication between Alice and Bob was $O(f(r))$ bits,
and hence we can solve the \textsc{Index} problem in $f(r)$ bits. Recall that the
lower bound for the \textsc{Index} problem is $\Omega(N)=\Omega(r^2)$.
Note that $|G_X|=n=2r+5=O(r)$, and hence $\Omega(r^2)=\Omega(n^2)$.
\qed
\end{proof}

\begin{corollary}
Let $\epsilon>0$ be a constant. Any (randomized) streaming algorithm
that achieves a $(\frac{7}{6}-\epsilon)$-approximation for a liar's dominating set requires $\Omega(n^2)$ space.
\label{corollary_str1}
\end{corollary}
\begin{proof}
Theorem~\ref{thm:lower-bound-lds}
shows that distinguishing between whether the minimum value of
the LDS is 6 or 7 requires $\Omega(n^2)$ bits. The claim follows since $6\cdot (\frac{7}{6}-\epsilon) <7$
\end{proof}

\subsection{Streaming Lower Bounds for $\boldmath{k}$-DS}\label{streaming}

\begin{theorem}
For any $k=O(1)$, any randomized (by randomized algorithm we mean that the algorithm
should succeed with probability $\geq \frac{2}{3}$) streaming algorithm
for the $k$-tuple dominating set problem on $n$-vertex graphs requires $\Omega(n^2)$ space.
\label{thm:lower-bound-kds}
\end{theorem}
\begin{proof}

%

We reduce from the \textsc{Index} problem. We construct an instance $G_X$
of $k$-tuple dominating set.
Assume that Alice has an algorithm that solves the $k$-tuple dominating set problem using $f(r)$ bits.
First, we insert the edges corresponding to the edge interpretation of
$X$ between nodes $v_i$ and $w_j$:
  for each $i, j \in [k]$, Alice adds the edge $(v_i, w_j)$ if
  the corresponding entry in $X$ is 1.
Then, Alice sends the memory contents of her algorithm to Bob, using
$f(r)$ bits.

Now, Bob has the index $i \in [N]$, which he interprets as $(I,J)$ under
the same bijection $\phi:[N]\rightarrow [r] \times [r]$.
Bob receives the memory contents of the algorithm, and does the following --- (1) adds (k+1) vertices $A=\{a_1, a_2, \ldots, a_k\}$ and $b$,
(2) adds edges $\{a_{i}-b\ :\ 1\leq i\leq k\}$,
(3) adds an edge from each vertex of $V\setminus (v_I)$ to each vertex of $A$, (4) adds an edge from each vertex of $V\setminus (w_J)$ to each vertex of $A$, (5) adds an edge $v_I$ to each vertex of $A\setminus a_k$,
(6) adds an edge from $w_J$ to each vertex of $A\setminus a_k$.


Now we show finding the minimum value of a $k$-tuple dominating set of $G_X$
allows us to solve the corresponding instance $X$ of \textsc{Index}.

\begin{lemma}
The minimum size of a $k$-tuple dominating set of $G_X$ is $k+1$ if and only if
$x_i = 1$.
\label{thm:lower-bound-kds-insertion}
\end{lemma}


\begin{proof}
Let $D$ be a minimum $k$-DS of $G_X$. Since $b$ has only $k$ neighbors in $G_X$, we can assume, w.l.o.g., $A\subseteq D$. Observe, every vertex in $G_X\setminus \{v_I , w_J\}$ already has $k$ neighbors in $A$. Both $v_I$ and $w_J$ have exactly $k-1$ neighbors in $A$. Hence, $|D|\geq k+1$.

Suppose that
$x_i = 1$, i.e., $v_{I}-w_J$ is an edge in $G_X$. We now claim that $A\cup v_I$ is a $k$-tuple dominating set of $G_X$.
We have observed above that each vertex in $G_X\setminus $ has $k$ neighbors in $A\subseteq D$.
So we just need to verify the condition for $v_I$ and $w_J$ now.
The claim follows since $N[v_I]\cap D = \{a_1, a_2, \ldots, a_{k-1}\}\cup v_I = N[w_J]\cap D$, and $|\{a_1, a_2, \ldots, a_{k-1}\}\cup v_I|=k$.

Now suppose that $x_i = 0$, i.e., $v_{I}$ and $w_J$ do not have an edge between them in $G_X$. Note that $A\cup v_I \cup w_J$ is indeed a $k$-tuple dominating set for $G_X$ of size $k+2$. We now claim that $G_X$ has no $k$-tuple dominating set of size $k+1$. Suppose to the contrary that there is a $k$-tuple dominating set $D'$ of $G_X$ of size $k+1$. Since $A$ has to be part of any minimum $k$-tuple dominating set, it follows that $D'=A\cup \beta$ for some vertex $\beta\in G_X \setminus A$. We now consider all choices for where the vertex $\beta$ can be chosen from (and derive a contradiction in each case):
\begin{itemize}
  \item \underline{$\beta=b$}: Then we have $|N[v_I]\cap D'|=k-1$
  \item \underline{$\beta\in V$}: Then we have $|N[v_I]\cap D'|=k-1$
  \item \underline{$\beta\in W$}: Then we have $|N[w_J]\cap D'|=k-1$
\end{itemize}
This completes the proof. 
~\qedsymbol
\end{proof}

Thus, by checking whether the value of minimum $k$-tuple dominating set on the instance
$G_X$ is $k+1$ or $k+2$, Bob can determine the index
$x_i$.
The total communication between Alice and Bob was $O(f(r))$ bits,
and hence we can solve the \textsc{Index} problem in $f(r)$ bits. Recall that the
lower bound for the \textsc{Index} problem is $\Omega(N)=\Omega(r^2)$. Note that $|G_X|=n=2r+k+1=O(r)$ since $k=O(1)$, and hence $\Omega(r^2)=\Omega(n^2)$.
\qed
\end{proof}

\begin{corollary}
Let $1>\epsilon>0$ be any constant. Any (randomized) streaming algorithm that approximates a $k$-tuple dominating set within a relative error of
$\epsilon$ requires $\Omega(n^2)$ space.
\label{corollary_str2}
\end{corollary}
\begin{proof}
Choose $\epsilon=\frac{1}{k}$. Theorem~\ref{thm:lower-bound-kds-insertion}
shows that the relative error is at most $\frac{1}{k+2}$, which is less than $\epsilon$.
Hence finding an approximation within $\epsilon$ relative error amounts to
finding the exact value of the $k$-tuple dominating set. Hence, the claim follows from the lower bound of $\Omega(n^2)$
of Theorem~\ref{thm:lower-bound-kds-insertion}.
\end{proof}

\subsection{W-Hardness Results for LDS}\label{hardness_param}
The LDS problem was shown to be NP-complete on general graphs by Slater in \cite{PJS09},
where the problem was introduced.
Later this problem was considered in \cite{PP13,RS09}, and was shown to be
NP-complete on bipartite graphs, split graphs and planar graphs.
In \cite{BGP17}, Bishnu et al. proved that the LDS problem on planar graphs admits
a linear kernel, and is W[2]-hard while considered on general graphs.
We study the LDS problem on bipartite graphs and show that it is 
W[2]-hard. 
Our approach is inspired by the W[2]-hardness results of \cite{RS06}.

\begin{theorem}\label{w2hard_bipartite}
Liar's dominating set on bipartite graphs is W[2]-hard.
\label{thm:w2-one}
\end{theorem}

\begin{proof}
We prove this by giving a parameterized reduction from the dominating set problem in general undirected graphs.
Let $(G=(V,E), k)$ be an instance of the dominating set, where $k$ denotes the size of the dominating set.
We construct a bipartite graph $G'=(V', E')$ from $G=(V,E)$.
First, we create two copies of $V$, namely $V_1=\{u_1|u\in V\}$ and $V_2=\{u_2|u\in V\}$.
Next, we introduce two extra vertices $z_1,z_2$ in $V_1$, and two extra vertices $z_{1}',z_{2}'$ in $V_2$.
Furthermore, we introduce two special vertices $s_{z'_1}, s_{z'_2}$ in $V_1$ and two special vertices $s_{z_1}, s_{z_2}$ in $V_2$.
The entire vertex set $V'$ is $V_1\cup V_2$, where $V_1=\{u_1|u\in V\}\cup \{z_1,z_2, s_{z'_1}, s_{z'_2}\}$ and 
$V_2=\{u_2|u\in V\}\cup\{z_{1}',z_{2}', s_{z_1}, s_{z_2}\}$.
Now, if there is an edge $(u,v)\in E$, then we draw the edges $(u_1,v_2)$ and $(v_1,u_2)$.
We draw the edges of the form $(u_1,u_2)$ in $G'$ for every vertex $u\in V$.
Then, we add edges from every vertex in $V_1\setminus \{z_1,z_2, s_{z'_1}, s_{z'_2}\}$ to $z_{1}',z_{2}'$,
and the edges from every vertex in $V_1\setminus \{z'_1,z'_2, s_{z_1}, s_{z_2}\}$ to $z_{1},z_{2}$.
Finally we add the edges $(z_1,z'_1), (z_2,z'_2)$ and $(z_1, s_{z_1}), (z_2, s_{z_2}), (z'_1, s_{z_1}), (z'_2, s_{z'_2})$.
This completes the construction (see Figure~\ref{fig:w2hard_bipartite}).

\begin{figure}[h]
 \centering
 \includegraphics[scale=0.5]{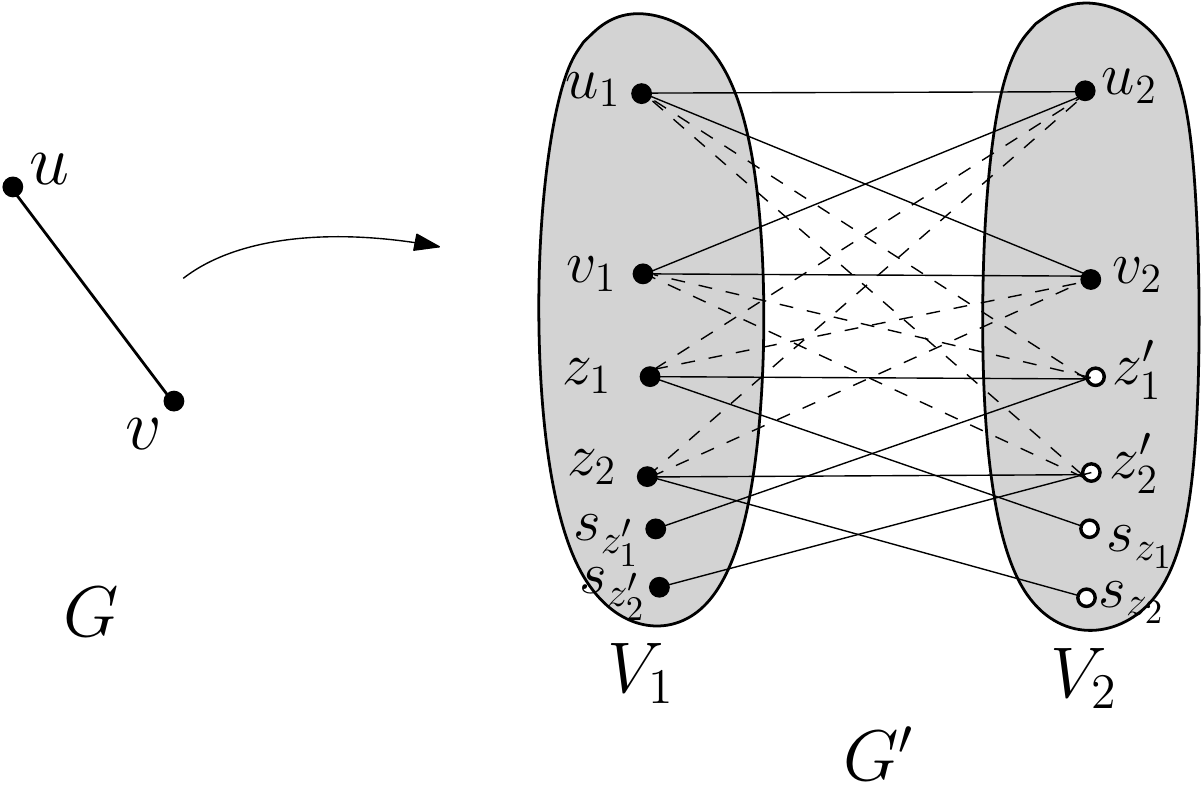}\vspace{-0.1in}
 \caption{Construction of $G'$ from $G$ (Illustration of Theorem \ref{w2hard_bipartite}).} \vspace{-0.1in}
 \label{fig:w2hard_bipartite}
 \end{figure}

We show that $G$ has a dominating set of size $k$ if and only if $G'$ has a LDS of size $k+8$.
Let $D$ denote the dominating set of the given graph $G$.
We claim that $D'=\{u_1|u\in D\}\cup\{z_1,z_2, s_{z'_1}, s_{z'_2}\} \cup \{z_{1}',z_{2}', s_{z_1}, s_{z_2}\}$ 
is a LDS of $G'$.
Note that for any vertex $v\in V'$, $|N_{G'}[v]\cap D'|\geq 2$, since 
$\{z_1,z_2,z_{1}',z_{2}'\}$ is in $D'$.
This fulfills the first condition of the LDS.
Now, for every pair of vertices, $u,v\in V'$, we show that $|(N_{G'}[u]\cup N_{G'}[v])\cap D'|\geq 3$.
If $u,v\in V_1$, we know $|(N_{G'}[u]\cup N_{G'}[v])\cap D'|\geq 2$ due to $z'_1,z'_2$.
Now, in the dominating set at least one additional vertex dominates them. 
Thus, $|(N_{G'}[u]\cup N_{G'}[v])\cap D'|\geq 3$.
Similarly, when $u,v\in V_2$ or $u\in V_1$ and $v\in V_2$.
This fulfills the second condition of the LDS.

Conversely, let $D'$ be a LDS in $G'$.
Note that, $\{z_1,z_2,z'_1,z'_2\}$ are always part of $D'$, since $z_1,z_2$ are the only neighbors of $s_{z_1}, s_{z_2}$ and
$z'_1,z'_2$ are the only neighbors of $s_{z'_1}, s_{z'_2}$.
These special vertices are taken in the construction to enforce $\{z_1,z_2,z'_1,z'_2\}$ to be in $D'$.
Now we know, for any pair of vertices $p,q$, $|(N_{G'}[p]\cup N_{G'}[q])\cap D'|\geq 3$.
This implies $p,q$ is dominated by at least one vertex or one of them is picked,
except $\{z_1,z_2,z'_1,z'_2\}$.
Otherwise, $|(N_{G'}[p]\cup N_{G'}[q])\cap D'|< 3$.
This violates the second condition of LDS.
Now, when $p,q$ are both part of the same edge in $G$ (say $u_2,v_2 \in V_2$ ; see Figure~\ref{fig:w2hard_bipartite}), 
we need at least one vertex from $\{u_1,v_1,u_2,v_2\}$ in $D'$.
This means that for every vertex $v\in V$, $|N_{G}[v]\cap D|\geq 1$.
Thus, $D$ is a dominating set of $G$ where the cardinality of $D$ is at most $k$.
\qed
\end{proof}

\subsection*{Acknowledgements}
The authors would like to thank Rajesh Chitnis and M. S. Ramanujan 
for interesting discussions during various stages of the research.

\bibliographystyle{abbrv}
\bibliography{ref_lds}
\end{document}